\documentclass[11pt,letterpaper]{article}
\usepackage[T1]{fontenc}
\usepackage{lmodern}
\usepackage[margin=1in]{geometry}
\usepackage{amsmath,amssymb,amsthm,mathtools}
\usepackage{microtype}
\usepackage[hidelinks]{hyperref}
\hypersetup{pdftitle={Continuity of Regularized Channel R\'enyi Divergences},pdfauthor={Jinzhao Wang and Yuxiang Yang},pdfsubject={Continuity proof using Stinespring representations}}
\numberwithin{equation}{section}
\newtheorem{theorem}{Theorem}
\newtheorem{lemma}[theorem]{Lemma}
\newtheorem{corollary}[theorem]{Corollary}
\newcommand{\N}{\mathcal N}
\newcommand{\M}{\mathcal M}
\newcommand{\id}{\operatorname{id}}
\newcommand{\Tr}{\operatorname{Tr}}
\newcommand{\supp}{\operatorname{supp}}
\newcommand{\reg}{\mathrm{reg}}
\newcommand{\cp}{\mathrm{CP}}
\newcommand{\Dt}{D}
\newcommand{\Qt}{Q}
\newcommand{\cL}{\mathcal L}
\newcommand{\cS}{\mathcal S}
\newcommand{\norm}[1]{\left\lVert#1\right\rVert}
\begin{document}
\begin{center}
{\LARGE Continuity of Regularized Channel R\'enyi Divergences\par}
\vspace{8pt}
{\large Jinzhao Wang\textsuperscript{1}\qquad Yuxiang Yang\textsuperscript{2}\par}
\vspace{5pt}
{\small
\textsuperscript{1}Zyphra, San Francisco, CA 94105, USA\\
\textsuperscript{2}QICI Quantum Information and Computation Initiative,\\
School of Computing and Data Science, The University of Hong Kong,\\
Pokfulam Road, Hong Kong, China\par}
\end{center}
\vspace{2pt}
\begin{abstract}
We prove that the regularized, stabilized sandwiched R\'enyi divergence of finite-dimensional quantum channels converges to their regularized relative entropy as the R\'enyi order tends to one. The key tool is the channel hockey-stick divergence: Gour's Stinespring approximation bound and a Schatten norm estimate amplify an asymptotic bound below one into exponential decay at higher threshold rates. For channel pairs with finite max-relative entropy, known operational connections then give exponential strong converses for parallel and adaptive discrimination, a sharp zero--one testing law, and the subchannel asymptotic equipartition property. We also provide a Lean formalization to certify our proof.
\end{abstract}

\section{Introduction}\label{sec:problem}
Suppose an unknown channel is either $\N$ or $\M$ and can be used $n$ times. Inputs may be entangled across uses, and adaptive protocols can process earlier outputs before choosing later inputs \cite{WBHK}. The type-I error is rejecting $\N$ when it is true; the type-II error is accepting $\N$ when $\M$ is true. Asymmetric discrimination asks how fast the type-II error can decay while the type-I error stays below a fixed $\epsilon\in(0,1)$. For state discrimination, quantum Stein's lemma identifies relative entropy as the optimal type-II error exponent \cite{HP,ON}. Channel relative entropy can be nonadditive~\cite{FFRS}, so its regularization is the natural candidate for the channel threshold. The strong converse asks whether rates above this threshold force acceptance under $\N$ to vanish; the exponential strong converse asks that it vanish exponentially. Earlier strong Stein results cover classical channels \cite{Hayashi}, classical--quantum channels \cite{WBHK}, and replacer alternatives \cite{CMW}, while the recent minimax study of Fang et al.~\cite{FCCGH} gives a general entangled-tester result with vanishing type-I error and fixed-error strong converses for restricted models.

Sandwiched R\'enyi divergences at orders $\alpha>1$ bound the tradeoff between these two errors exponentially \cite{MLDFT,WWY}. A channel R\'enyi chain rule extends the bound to adaptive protocols \cite{FF}. To recover the desired regularized relative-entropy threshold by taking $\alpha\downarrow1$, one needs continuity of the regularized channel R\'enyi divergence at order one \cite{FF,FGW}. Continuity for each fixed input and block length does not settle this: the channel divergence is optimized over entangled inputs and arbitrarily large blocks. We prove that the stabilized, regularized divergences defined in Section~\ref{sec:prelim} have the required limit.

\begin{theorem}[Continuity of the sandwiched divergence at order one]\label{thm:main}
For any two finite-dimensional quantum channels $\N,\M$ with the same input and output spaces, the following equality holds:
\begin{equation}\label{eq:main}
 \lim_{\alpha\to1}\Dt_\alpha^{\reg}(\N\|\M)
       =D^{\reg}(\N\|\M).
\end{equation}
\end{theorem}

The operational implications of continuity were identified in earlier work. Fawzi and Fawzi proved a channel R\'enyi chain rule and characterized the adaptive strong-converse exponent, leaving right continuity at one open \cite{FF}. Fang, Gour, and Wang related this continuity to the exponential strong converse for parallel and adaptive discrimination~\cite{FGW}. Gour established equivalent formulations through the hockey-stick testing threshold, the channel Lorenz divergence, and an asymptotic equipartition property (AEP) with subchannel smoothing\footnote{Gour also showed that the corresponding channel AEP with trace-preserving smoothing fails in general~\cite{Gour}.} \cite[Theorem 14]{Gour}.

We establish an amplification principle for the channel hockey-stick divergence \cite{Gour,SW}: an asymptotic bound strictly below one at a threshold rate forces exponential decay at every higher rate. We adapt Gour's Stinespring filter \cite{Gour} to a prescribed isometry and environment, obtaining an operator approximation with error at most the square root of the hockey-stick divergence, uniformly over inputs. We combine two such approximations through a Schatten norm estimate in Section~\ref{sec:proof}; Section~\ref{sec:consequences} surveys the applications in channel discrimination and channel AEP.

Three concurrent works also prove continuity at order one. Gao, Ji, and Liu \cite{GaoJiLiu} obtain a uniform Stinespring approximation through iterative rate reduction. Schmitt and Sutter \cite{SchmittSutter} use a variational formula and spectral confinement. Zhu and Wang \cite{ZhuWang} use a closely related two-rate, three-amplitude argument, constructing their common approximation directly from a positive Choi slack; they also treat general process testers, including indefinite causal order.\par\medskip

\noindent\textbf{AI usage.} 
The proof presented in this paper was proposed by GPT-6 Astra. AI tools also assisted with drafting and revising the manuscript. The authors take full responsibility for the final content. A Lean formalization of Theorem~\ref{thm:main} is available online~\cite{WangLean}. Comparator checked its theorem statement and replayed the proofs in the Lean kernel.

\subsection{Preliminaries}\label{sec:prelim}
All spaces are finite-dimensional and all logarithms are base two. Write $\cL(H)$ for linear operators, $\cS(H)$ for density operators, $\norm{\cdot}_p$ for the Schatten norms, and $\norm{\cdot}_\diamond$ for the diamond norm. A star on a map ${}^*$ denotes its Hilbert--Schmidt adjoint. For a unit vector $|\psi\rangle$, write $\psi=|\psi\rangle\!\langle\psi|$. We use
\begin{align}
 D(\rho\|\sigma)&=\Tr\rho(\log\rho-\log\sigma),\label{eq:stateD}\\
 \Dt_\alpha(\rho\|\sigma)
 &=\frac1{\alpha-1}\log\Tr\!\left[
 \left(\sigma^{\frac{1-\alpha}{2\alpha}}\rho\,
              \sigma^{\frac{1-\alpha}{2\alpha}}\right)^\alpha\right],
 \qquad \alpha\in[1/2,\infty)\setminus\{1\}.
 \label{eq:stateRenyi}
\end{align}
We use the usual support conventions: negative powers act on $\supp\sigma$, and $D,\Dt_{\alpha>1}$ are infinite unless $\supp\rho\subseteq\supp\sigma$. For $\alpha<1$, a zero trace in \eqref{eq:stateRenyi} gives $+\infty$. We use tensor additivity, monotonicity in $\alpha$, and the limits at one and infinity \cite{MLDFT}, together with data processing \cite{FL,Beigi}.

For channels $\N,\M:\cL(A)\to\cL(B)$, $\mathbb D\in\{D,\Dt_\alpha\}$, the stabilized and regularized divergences are
\begin{align}
 \mathbb D(\N\|\M)
 &=\sup_{\substack{|\psi\rangle\in R\otimes A\\\langle\psi|\psi\rangle=1}}
 \mathbb D\bigl((\id_R\otimes\N)(\psi)\|
                     (\id_R\otimes\M)(\psi)\bigr),
 \qquad R\simeq A,\label{eq:channel}\\
 \mathbb D^{\reg}(\N\|\M)
 &=\lim_{n\to\infty}\frac1n\mathbb D(\N^{\otimes n}\|\M^{\otimes n})
 =\sup_{n\ge1}\frac1n\mathbb D(\N^{\otimes n}\|\M^{\otimes n}).
 \label{eq:regularized}
\end{align}
Pure inputs with $R\simeq A$ suffice: purification cannot decrease the divergence by data processing, and Schmidt decomposition compresses a pure input's reference to dimension $\dim A$. Tensor additivity on product inputs makes the channel divergence superadditive in the number of uses, so Fekete's lemma gives \eqref{eq:regularized}. For an $n$-use block, the reference may have dimension $(\dim A)^n$, and the input may be entangled across all uses.

For two completely positive maps $\mathcal A$ and $\mathcal B$ with the same input and output spaces, write $\mathcal A\le_{\cp}\mathcal B$ when $\mathcal B-\mathcal A$ is completely positive, equivalently $J_{\mathcal A}\le J_{\mathcal B}$ \cite{Choi}, using the unnormalized Choi convention
\begin{equation}\label{eq:choi}
 J_{\mathcal A}=(\id_{A'}\otimes\mathcal A)
                  (|\Omega_A\rangle\!\langle\Omega_A|),
 \qquad |\Omega_A\rangle=\sum_i|i\rangle_{A'}|i\rangle_A.
\end{equation}
We use the max-relative entropy \cite{Datta} and its CP-order extension to maps \cite[Eq.~(13)]{Gour}:
\begin{equation}\label{eq:Dmax}
 \begin{split}
 D_{\max}(\rho\|\sigma)&=\inf\{\log\lambda:\lambda>0,\ \rho\le\lambda\sigma\},\\
 D_{\max}(\mathcal A\|\mathcal B)&=\inf\{\log\lambda:\lambda>0,\ \mathcal A\le_{\cp}\lambda\mathcal B\},
 \end{split}
\end{equation}
with $\inf\varnothing=+\infty$.
For a linear operator $V:A\to B\otimes E$, write
\begin{equation}\label{eq:Phi}
 \Phi_V(X)=\Tr_E(VXV^\dagger).
\end{equation}
A Stinespring isometry of a channel $\N$ is an operator $V$ with $\Phi_V=\N$ and $V^\dagger V=I_A$.
An operator sequence $U_n$ has \emph{CP rate at most $r$} relative to $\M$ if
$\Phi_{U_n}\le_{\cp}C\,2^{nr}\M^{\otimes n}$ for all sufficiently large $n$ and some constant $C$ independent of $n$.

\section{Proof}\label{sec:proof}
For the left continuity, let $\psi$ be a block input and let $\rho_{n,\psi}=(\id\otimes\N^{\otimes n})(\psi)$ and $\sigma_{n,\psi}=(\id\otimes\M^{\otimes n})(\psi)$. Monotonicity in $\alpha$, state-level continuity at one, and interchange of the two suprema give\footnote{This state-level limit includes support mismatch: if $\supp\rho\not\subseteq\supp\sigma$, the trace in \eqref{eq:stateRenyi} tends to $\Tr(\Pi_{\supp\sigma}\rho)<1$ as $\alpha\uparrow1$, so $\Dt_\alpha(\rho\|\sigma)\to+\infty$.}
\begin{align*}
 \lim_{\alpha\uparrow1}\Dt_\alpha^{\reg}(\N\|\M)
 &=\sup_{1/2\le\alpha<1}\sup_{n,\psi}
       \frac1n\Dt_\alpha(\rho_{n,\psi}\|\sigma_{n,\psi})\\
 &=\sup_{n,\psi}\sup_{1/2\le\alpha<1}
       \frac1n\Dt_\alpha(\rho_{n,\psi}\|\sigma_{n,\psi})
 =D^{\reg}(\N\|\M).
\end{align*}

For right continuity, first let $d_\infty=D_{\max}(\N\|\M)$. If $d_\infty=+\infty$, then $\supp J_\N\not\subseteq\supp J_\M$, so a normalized maximally entangled input gives infinite relative entropy and infinite order-$\alpha>1$ divergence. Thus right continuity is immediate in this case.

For the remainder of the proof, assume $d_\infty<\infty$. CP order and  $\Dt_\alpha\le D_{\max}$ (for states) give
\begin{equation}\label{eq:cap}
 \N^{\otimes n}\le_{\cp}2^{n d_\infty}\M^{\otimes n},\qquad
 0\le D^{\reg}(\N\|\M)
 \le\Dt_\alpha^{\reg}(\N\|\M)\le d_\infty\quad(\alpha>1).
\end{equation}
Abbreviate
\begin{equation}\label{eq:thresholds}
 d=D^{\reg}(\N\|\M),\qquad
 d_\alpha=\Dt_\alpha^{\reg}(\N\|\M),\qquad
 d_+=\inf_{\alpha>1}d_\alpha=\lim_{\alpha\downarrow1}d_\alpha.
\end{equation}
We want to prove $d_+= d$, controlling the infimum over $\alpha$ together with the supremum over $n$.

\subsection{Proof sketch}\label{sec:sketch}
We know $d\le d_+$ from \eqref{eq:cap}. Suppose instead that $d<d_+$ and choose $d<r<d_+$. The relative-entropy testing bound \eqref{eq:weak} makes the channel hockey-stick divergence $E_{2^{nr}}(\N^{\otimes n}\|\M^{\otimes n})$ stay asymptotically below one. Write $b=\limsup_n\sqrt{E_{2^{nr}}(\N^{\otimes n}\|\M^{\otimes n})}<1$. In contrast, the R\'enyi testing bound \eqref{eq:high} makes the same divergence vanish at every rate $\ell>d_+$. The definition of $E_\lambda$ appears in \eqref{eq:hchannel} below.

Fix a Stinespring isometry $V_n$ of $\N^{\otimes n}$. Lemma~\ref{lem:filter} gives an approximation $X_n$ at rate $r$ with $\limsup_n\norm{V_n-X_n}_\infty\le b$. Applying it at any rate $\ell>d_+$ gives $W_n$ with $\norm{V_n-W_n}_\infty\to0$. Both use the same environment, so
\[
 \underbrace{V_n}_{\text{exact operator}}
 =\underbrace{X_n}_{\text{cheaper approximation}}
 +\underbrace{(W_n-X_n)}_{\text{correction}}
 +\underbrace{(V_n-W_n)}_{\text{vanishing leftover}}.
\]
The first piece has CP rate at most $r$. The correction has a CP bound at rate $\ell$, up to a fixed multiplier, and norm at most $b+o(1)<1$; the leftover has vanishing norm and a fixed finite CP rate because $d_\infty<\infty$. For each fixed $n$, Lemma~\ref{lem:decomposition} bounds arbitrarily many copies of this $n$-use block, even for inputs entangled across all copies; letting the number of copies grow bounds the regularized divergence $d_\alpha$. In the contradiction argument we set $\ell=d_++t^{-2}$ and $\alpha=n/(n-t)$. We first fix $t$ and let $n\to\infty$, which removes the vanishing leftover. Then $t\to\infty$ suppresses the cheaper piece and brings the correction rate down to $d_+$. Since $d_\alpha\ge d_+$, the resulting inequality is $1\le b$, contradicting $b<1$. Thus $d<d_+$ is impossible. The same argument amplifies any asymptotic testing bound below one into decay at higher rates.

\subsection{Approximation of Stinespring operators}\label{sec:filter}
For $\lambda\ge1$, the state hockey-stick divergence \cite{SW} and its stabilized channel extension \cite{Gour} are
\begin{align}
 E_\lambda(\rho\|\sigma)
 &=\Tr(\rho-\lambda\sigma)_+
 =\max_{0\le T\le I}\Tr[T(\rho-\lambda\sigma)],\label{eq:hstate}\\
 E_\lambda(\N\|\M)
 &=\sup_{\substack{|\psi\rangle\in R\otimes A\\\langle\psi|\psi\rangle=1}}
 E_\lambda\bigl((\id_R\otimes\N)(\psi)\|
                    (\id_R\otimes\M)(\psi)\bigr),\qquad R\simeq A.
 \label{eq:hchannel}
\end{align}
Here $X_+$ is the positive part of a Hermitian operator. For output states $\rho,\sigma$ from a common channel input and a test $0\le T\le I$, let $p=\Tr(T\rho)$ and $\beta=\Tr(T\sigma)$. The type-I and type-II errors are $1-p$ and $\beta$, respectively. Thus $E_\lambda$ maximizes $p-\lambda\beta$ over inputs and tests, and $0\le E_\lambda\le1$.

We use Gour's filter \cite[Eqs.~(44)--(49)]{Gour} in a fixed-dilation form, in which the chosen isometry $V$ and environment $E$ are fixed for all $\lambda$. This allows us to compare approximations at different~$\lambda$.

\begin{samepage}
\begin{lemma}[Gour's filter in a fixed dilation]\label{lem:filter}
Let $\N,\M:\cL(A)\to\cL(B)$ be quantum channels, let $\lambda\ge1$, and fix a Stinespring isometry $V:A\to B\otimes E$ of $\N$. There is a linear operator\footnote{The operator $W_\lambda$ need not be an isometry or a contraction; $\Phi_{W_\lambda}$ need not preserve or decrease trace. A contraction is a linear operator $T$ with $\norm{T}_\infty\le1$.} $W_\lambda:A\to B\otimes E$ such that
\begin{equation}\label{eq:filter}
 \Phi_{W_\lambda}\le_{\cp}\lambda\M,\qquad
 \norm{V-W_\lambda}_\infty\le\sqrt{E_\lambda(\N\|\M)}.
\end{equation}
Since $\norm V_\infty=1$, the triangle inequality also gives
$\norm{W_\lambda}_\infty\le1+\sqrt{E_\lambda(\N\|\M)}$.
The same fixed isometry $V$ and environment $E$ may be used for every $\lambda$.
\end{lemma}
\end{samepage}

\begin{proof}
Gour's construction \cite[proof of Theorem~2, Eqs.~(44)--(49)]{Gour} applies for every $\lambda\ge1$: take $\gamma=\lambda$ and $s=E_\lambda(\N\|\M)$ in those steps.\footnote{The CP Radon--Nikodym theorem assumes a minimal dilation. Gour's direct-sum dilation may not be minimal; restrict it to its minimal environmental subspace and extend the resulting effect by zero.} The construction gives a Stinespring isometry $\widetilde V:A\to B\otimes\widetilde E$ of $\N$ and operators $X,Y$ on the same space with
\[
 \widetilde V=X+Y,\qquad
 \Phi_X\le_{\cp}\lambda\M,\qquad
 \norm{Y}_\infty\le\sqrt{E_\lambda(\N\|\M)}.
\]
Using $|\Omega_A\rangle$ from \eqref{eq:choi}, the vectors $(I_{A'}\otimes\widetilde V)|\Omega_A\rangle$ and $(I_{A'}\otimes V)|\Omega_A\rangle$ both purify $J_\N$. Purification uniqueness therefore gives a partial isometry $U:\widetilde E\to E$ such that
\[
 (I_{A'B}\otimes U)(I_{A'}\otimes\widetilde V)|\Omega_A\rangle
 =(I_{A'}\otimes V)|\Omega_A\rangle.
\]
The map $L\mapsto(I_{A'}\otimes L)|\Omega_A\rangle$ is injective, so $(I_B\otimes U)\widetilde V=V$. Set $W_\lambda=(I_B\otimes U)X$. Since $U^\dagger U\le I_{\widetilde E}$,
\begin{align*}
 \Phi_X-\Phi_{W_\lambda}
 &=\Phi_{(I_B\otimes(I_{\widetilde E}-U^\dagger U)^{1/2})X}\ge_{\cp}0,\\
 \Phi_{W_\lambda}&\le_{\cp}\Phi_X\le_{\cp}\lambda\M,
 &\norm{V-W_\lambda}_\infty
 &=\norm{(I_B\otimes U)Y}_\infty
 \le\norm Y_\infty\le\sqrt{E_\lambda(\N\|\M)}.
\end{align*}
The transfer may depend on $\lambda$, but $V$ and $E$ do not.
\end{proof}

\subsection{Schatten norm bound}\label{sec:decomposition}
We next bound the regularized R\'enyi divergence using a Stinespring decomposition.

\begin{lemma}[Schatten norm bound]\label{lem:decomposition}
Let $\N,\M:\cL(A)\to\cL(B)$ be quantum channels, let $n,q$ be positive integers, and let $V:A^{\otimes n}\to B^{\otimes n}\otimes E$ be a Stinespring isometry of $\N^{\otimes n}$. Suppose there exist linear operators $U_i:A^{\otimes n}\to B^{\otimes n}\otimes E$, $i=1,\ldots,q$, and constants $b_i\ge0$, $\lambda_i>0$, such that
\[
 V=\sum_{i=1}^q U_i,\qquad
 \norm{U_i}_\infty\le b_i,\qquad
 \Phi_{U_i}\le_{\cp}\lambda_i\M^{\otimes n}.
\]
The operators $U_i$ need not be isometries or contractions. For every $\alpha>1$, with $s=(\alpha-1)/\alpha$,
\begin{equation}\label{eq:decomposition}
 \Dt_\alpha^{\reg}(\N\|\M)
 \le\frac{2}{ns}\log\left(\sum_{i=1}^q
               b_i^{1-s}\lambda_i^{s/2}\right).
\end{equation}
\end{lemma}

\begin{proof}
For a positive operator $\rho\ge0$, not necessarily normalized, and a state $\sigma$, let
\begin{equation}\label{eq:Q}
 \Qt_\alpha(\rho\|\sigma)
 =\Tr\left[\left(\sigma^{-s/2}\rho\sigma^{-s/2}\right)^\alpha\right],
 \qquad s=\frac{\alpha-1}{\alpha},
\end{equation}
with the usual support convention and $\Qt_\alpha(0\|\sigma)=0$. Homogeneity and the normalized-state bound $\Dt_\alpha\le D_{\max}$ imply, for $\rho\ne0$,
\begin{equation}\label{eq:Qbound}
 \rho\le c\sigma\quad\Longrightarrow\quad
 \Qt_\alpha(\rho\|\sigma)
 \le(\Tr\rho)^\alpha\left(\frac{c}{\Tr\rho}\right)^{\alpha-1}
 =c^{\alpha-1}\Tr\rho.
\end{equation}

Fix a positive integer $m$ and a pure input $\psi$ on $R\otimes A^{\otimes nm}$, with any finite reference $R$, and write
\[
 \rho=(\id_R\otimes\N^{\otimes nm})(\psi),\qquad
 \sigma=(\id_R\otimes\M^{\otimes nm})(\psi).
\]
Expand $V^{\otimes m}=\sum_{\mathbf i}U_{\mathbf i}$, where
$\mathbf i=(i_1,\ldots,i_m)\in\{1,\ldots,q\}^m$ and
$U_{\mathbf i}=U_{i_1}\otimes\cdots\otimes U_{i_m}$.
Reorder the output factors as $B^{\otimes nm}\otimes E^{\otimes m}$.
CP order is preserved under tensor products, and the operator norm is multiplicative. Therefore
\begin{equation}\label{eq:wordbounds}
 \rho_{\mathbf i}:=(\id_R\otimes\Phi_{U_{\mathbf i}})(\psi)
 \le\left(\prod_{j=1}^m\lambda_{i_j}\right)\sigma,
 \qquad
 \Tr \rho_{\mathbf i}\le\prod_{j=1}^m b_{i_j}^2.
\end{equation}
For the trace bound, $\Tr\rho_{\mathbf i}=\norm{(I_R\otimes U_{\mathbf i})|\psi\rangle}^2\le\norm{U_{\mathbf i}}_\infty^2\le\prod_j b_{i_j}^2$.

Choose common orthonormal bases $\{|a\rangle\}$ of $R\otimes B^{\otimes nm}$ and $\{|e\rangle\}$ of $E^{\otimes m}$. Let $C_{\mathbf i}:E^{\otimes m}\to R\otimes B^{\otimes nm}$ be the coefficient matrix of $|v_{\mathbf i}\rangle=(I_R\otimes U_{\mathbf i})|\psi\rangle$:
\[
 |v_{\mathbf i}\rangle=\sum_{a,e}c^{(\mathbf i)}_{ae}|a\rangle\otimes|e\rangle,
 \qquad C_{\mathbf i}=\sum_{a,e}c^{(\mathbf i)}_{ae}|a\rangle\!\langle e|.
\]
Taking partial traces of the component and full output vectors gives
\[
 \rho_{\mathbf i}=C_{\mathbf i}C_{\mathbf i}^\dagger,
 \qquad \rho=\left(\sum_{\mathbf i}C_{\mathbf i}\right)
       \left(\sum_{\mathbf i}C_{\mathbf i}\right)^\dagger.
\]
Thus the sum retains all cross terms. Since
\(
 \norm{\sigma^{-s/2}C_{\mathbf i}}_{2\alpha}
 =\Qt_\alpha(\rho_{\mathbf i}\|\sigma)^{1/(2\alpha)},
\)
equations \eqref{eq:Qbound}--\eqref{eq:wordbounds} and the Schatten triangle inequality yield
\begin{equation}\label{eq:triangle}
 \Qt_\alpha(\rho\|\sigma)^{1/(2\alpha)}
 =\left\lVert\sum_{\mathbf i}\sigma^{-s/2}C_{\mathbf i}\right\rVert_{2\alpha}
 \le\sum_{\mathbf i}\norm{\sigma^{-s/2}C_{\mathbf i}}_{2\alpha}
 \le\sum_{\mathbf i}\prod_{j=1}^m b_{i_j}^{1-s}\lambda_{i_j}^{s/2}
 =\left(\sum_{i=1}^q b_i^{1-s}\lambda_i^{s/2}\right)^m.
\end{equation}
The last equality factors the sum of the product bounds over $\mathbf i=(i_1,\ldots,i_m)$.
All inverse powers act on $\supp\sigma$, which contains the ranges of the coefficient matrices by \eqref{eq:wordbounds}.

Taking logarithms and optimizing the input gives
\[
 \frac1{nm}\Dt_\alpha(\N^{\otimes nm}\|\M^{\otimes nm})
 \le\frac{2}{ns}\log\left(\sum_{i=1}^q b_i^{1-s}\lambda_i^{s/2}\right).
\]
Taking $m\to\infty$ at fixed $n$ and $\alpha$ gives \eqref{eq:decomposition}: the subsequence of block lengths $nm$ has the same regularized limit as in \eqref{eq:regularized}.
\end{proof}

\subsection{Testing bounds and conclusion}\label{sec:conclusion}
The next bounds keep the hockey-stick divergence asymptotically below one at rates above $d$ and make it vanish at rates above $d_+$. We then show that these thresholds coincide.

\begin{lemma}[Relative-entropy and R\'enyi testing bounds]\label{lem:testing}
Let $\N,\M:\cL(A)\to\cL(B)$ be quantum channels with $D_{\max}(\N\|\M)<\infty$. For every positive integer $n$ and every $r>0$,
\begin{equation}\label{eq:weak}
 E_{2^{nr}}(\N^{\otimes n}\|\M^{\otimes n})
 \le\frac{D^{\reg}(\N\|\M)}r+\frac1{nr}.
\end{equation}
For every $\alpha>1$, every $\ell>\Dt_\alpha^{\reg}(\N\|\M)$, and every positive integer $n$,
\begin{equation}\label{eq:high}
 E_{2^{n\ell}}(\N^{\otimes n}\|\M^{\otimes n})
 \le2^{-n(\alpha-1)[\ell-\Dt_\alpha^{\reg}(\N\|\M)]}.
\end{equation}
\end{lemma}

\begin{proof}
Tests with nonpositive objective satisfy both upper bounds automatically. Fix an input and test with $p-2^{nr}\beta>0$, and let $\rho,\sigma$ be its two output states. Since $p>0$ and $p\le2^{n d_\infty}\beta$ by \eqref{eq:cap}, we also have $\beta>0$.

Write $h_2(p)=-p\log p-(1-p)\log(1-p)$ for binary entropy, with $0\log0=0$. Data processing under the binary measurement gives
\begin{align*}
 D(\rho\|\sigma)
 &\ge p\log\frac p\beta+(1-p)\log\frac{1-p}{1-\beta}\\
 &=-h_2(p)-p\log\beta-(1-p)\log(1-\beta)
 \ge nr\,p-1.
\end{align*}
Here $\beta<p2^{-nr}\le2^{-nr}$, $h_2(p)\le1$, and $-(1-p)\log(1-\beta)\ge0$. On the other hand, \eqref{eq:channel}--\eqref{eq:regularized} give
\[
 D(\rho\|\sigma)\le D(\N^{\otimes n}\|\M^{\otimes n})\le nD^{\reg}(\N\|\M).
\]
Thus $p\le D^{\reg}(\N\|\M)/r+1/(nr)$. Since the testing objective is at most $p$, optimizing proves \eqref{eq:weak}.

For a test with $p-2^{n\ell}\beta>0$, we again have $p,\beta>0$ and $\beta<p2^{-n\ell}$. Since $1-\alpha<0$, this implies $p^\alpha\beta^{1-\alpha}\ge p\,2^{n\ell(\alpha-1)}$. Data processing for the R\'enyi divergence, followed by dropping the nonnegative term for the other outcome, gives
\[
 n\Dt_\alpha^{\reg}(\N\|\M)
 \ge\Dt_\alpha(\rho\|\sigma)
 \ge\frac1{\alpha-1}\log(p^\alpha\beta^{1-\alpha})
 \ge n\ell+\frac1{\alpha-1}\log p.
\]
Rearranging gives $p\le2^{-n(\alpha-1)[\ell-\Dt_\alpha^{\reg}(\N\|\M)]}$. Using $p-2^{n\ell}\beta\le p$ and optimizing proves~\eqref{eq:high}.
\end{proof}

Since $d_+=\inf_{\alpha>1}d_\alpha$, each $\ell>d_+$ admits a fixed $\alpha>1$ with $d_\alpha<\ell$. Lemma~\ref{lem:testing}~\eqref{eq:high} therefore implies, for every $\ell>d_+$,
\begin{equation}\label{eq:highvanishing}
 E_{2^{n\ell}}(\N^{\otimes n}\|\M^{\otimes n})\longrightarrow0
 \qquad(n\to\infty).
\end{equation}

\begin{proof}[Proof of Theorem~\ref{thm:main}]
We first establish the implication
\begin{equation}\label{eq:thresholdcriterion}
 \limsup_{n\to\infty}E_{2^{nr}}(\N^{\otimes n}\|\M^{\otimes n})<1
 \quad\Longrightarrow\quad d_+\le r\qquad(r\ge0).
\end{equation}
Fix $r\ge0$ satisfying the testing condition and assume, for a contradiction, that $r<d_+$. Let
\[
 b_n=\sqrt{E_{2^{nr}}(\N^{\otimes n}\|\M^{\otimes n})},
 \qquad b=\limsup_{n\to\infty}b_n.
\]
The testing condition gives $b<1$. Fix any $t\ge1$. Use $d_++t^{-2}$ as the higher approximation rate and let
\[
 \varepsilon_n=\sqrt{E_{2^{n(d_++t^{-2})}}(\N^{\otimes n}\|\M^{\otimes n})}.
\]
For this fixed $t$, \eqref{eq:highvanishing} gives $\varepsilon_n\to0$. Choose a Stinespring isometry $V_n$ of $\N^{\otimes n}$. Lemma~\ref{lem:filter}, applied twice to this same isometry, supplies operators $X_n,W_n$ with
\begin{align}
 \Phi_{X_n}&\le_{\cp}2^{nr}\M^{\otimes n},
 &\norm{V_n-X_n}_\infty&\le b_n,\label{eq:lowapprox}\\
 \Phi_{W_n}&\le_{\cp}2^{n(d_++t^{-2})}\M^{\otimes n},
 &\norm{V_n-W_n}_\infty&\le\varepsilon_n.
 \label{eq:highapprox}
\end{align}
Define $Y_n=W_n-X_n$ and $Z_n=V_n-W_n$, so $V_n=X_n+Y_n+Z_n$. The triangle inequality gives
\[
 \norm{X_n}_\infty\le1+b_n,\qquad
 \norm{Y_n}_\infty\le b_n+\varepsilon_n,\qquad
 \norm{Z_n}_\infty\le\varepsilon_n.
\]
The inequality $\Phi_{U-W}\le_{\cp}2\Phi_U+2\Phi_W$,\footnote{This follows from $2\Phi_U+2\Phi_W-\Phi_{U-W}=\Phi_{U+W}\ge_{\cp}0$.} together with $r<d_+\le d_\infty$ and $t\ge1$, gives
\begin{equation}\label{eq:threecaps}
 \Phi_{Y_n}\le_{\cp}4\,2^{n(d_++t^{-2})}\M^{\otimes n},
 \qquad
 \Phi_{Z_n}\le_{\cp}4\,2^{n(d_\infty+1)}\M^{\otimes n}.
\end{equation}
For the second bound, use $\Phi_{V_n}=\N^{\otimes n}\le_{\cp}2^{nd_\infty}\M^{\otimes n}$ and $d_++t^{-2}\le d_\infty+1$.
For integers $n>t$, apply Lemma~\ref{lem:decomposition} with $s=t/n$, equivalently $\alpha=n/(n-t)>1$. Multiplying \eqref{eq:decomposition} by $ns/2=t/2$ and exponentiating gives
\[
 2^{td_\alpha/2}\le(1+b_n)^{1-t/n}\,2^{tr/2}
 +2^{t/n}(b_n+\varepsilon_n)^{1-t/n}\,2^{t(d_++t^{-2})/2}
 +2^{t/n}\varepsilon_n^{1-t/n}\,2^{t(d_\infty+1)/2}.
\]
Using $d_+\le d_\alpha$ and dividing by $2^{td_+/2}$ now gives
\begin{equation}\label{eq:threepiecebound}
 1\le\frac{2^{td_\alpha/2}}{2^{td_+/2}}
 \le(1+b_n)^{1-t/n}\,2^{-t(d_+-r)/2}
 +2^{t/n}(b_n+\varepsilon_n)^{1-t/n}\,2^{1/(2t)}
 +2^{t/n}\varepsilon_n^{1-t/n}\,2^{t(d_\infty+1-d_+)/2}.
\end{equation}
Keep $t$ fixed and let $n\to\infty$. With $\delta_n=t/n$, weighted AM--GM gives $x^{1-\delta_n}\le(1-\delta_n)x+\delta_n\le x+\delta_n$ for every $x\ge0$. Applying this to $1+b_n$, $b_n+\varepsilon_n$, and $\varepsilon_n$ shows that their powered factors have upper limits at most $1+b$, $b$, and $0$, respectively. Since $2^{t/n}\to1$, taking the upper limit in \eqref{eq:threepiecebound} gives
\begin{equation}\label{eq:contradiction}
 1\le(1+b)\,2^{-t(d_+-r)/2}+b\,2^{1/(2t)}
 \qquad(t\ge1).
\end{equation}
Now let $t\to\infty$. Since $r<d_+$, the right-hand side tends to $b<1$, a contradiction. Thus $d_+\le r$. Taking $n\to\infty$ first avoids any need for uniform convergence in $t$.

To complete the proof, we show that $d_+=d$. By \eqref{eq:weak}, for every $r>d$,
\[
 \limsup_{n\to\infty}E_{2^{nr}}(\N^{\otimes n}\|\M^{\otimes n})
 \le\frac dr<1.
\]
By \eqref{eq:thresholdcriterion}, $d_+\le r$ for every $r>d$. Letting $r\downarrow d$ and using $d\le d_+$ proves $d_+=d$.
\end{proof}

\section{Applications}\label{sec:consequences}
We now survey some applications developed in earlier works: Fawzi and Fawzi established the adaptive R\'enyi bound \cite{FF}, Fang, Gour, and Wang related continuity to channel discrimination \cite{FGW}, and Gour established the equivalence of the parallel testing threshold, subchannel AEP, and Lorenz AEP \cite[Theorem 14]{Gour}\footnote{Gour uses the amortized relative entropy, which equals our $D^{\reg}(\N\|\M)$ by \cite[Corollary 3.7]{FFRS}.}. We apply Theorem~\ref{thm:main} to these known connections, giving short proofs of the resulting channel-discrimination, subchannel-smoothing, and sharp hockey-stick threshold statements and recording explicit exponential bounds. Throughout this section we assume $D_{\max}(\N\|\M)<\infty$.

\subsection{Channel discrimination}
The goal of channel discrimination is to distinguish between two channels $\N,\M: A\to B$ given $n$ uses of the channel.
Two classes of testing protocols are often compared. A \emph{parallel} protocol prepares a joint input on $R\otimes A^{\otimes n}$ and makes a joint binary measurement on $R\otimes B^{\otimes n}$ after all uses, without feedback between uses. An \emph{adaptive} protocol uses the channel sequentially, processing each output together with retained quantum memory to prepare the next input. Parallel protocols are included as the special case in which later inputs do not depend on earlier outputs. We thus focus on the adaptive case. Let $p_n$ ($\beta_n$) be the probability that the final test yields $\N$ when the actual channel is $\N$ ($\M$). Note that $\beta_n$ is the type-II error (false negative).

We recall the adaptive testing bound of~\cite{WBHK}, expressed through the regularized divergence in~\cite{FF}. The sandwiched R\'enyi chain rule~\cite{FF} gives, for $\alpha>1$ and any states $\rho,\sigma$ on $R\otimes A$,
\[
 \Dt_\alpha\bigl((\id_R\otimes\N)(\rho)\|
                       (\id_R\otimes\M)(\sigma)\bigr)
 \le \Dt_\alpha(\rho\|\sigma)+\Dt_\alpha^{\reg}(\N\|\M).
\]
The common initial state has divergence zero. Each channel use adds at most $\Dt_\alpha^{\reg}(\N\|\M)$, while processing between uses cannot increase the divergence. Let $\rho_n,\sigma_n$ be the final states under $\N,\M$; their divergence is therefore at most $n\Dt_\alpha^{\reg}(\N\|\M)$. For $p_n,\beta_n>0$, data processing under the final binary test gives
\[
 n\Dt_\alpha^{\reg}(\N\|\M)
 \ge\Dt_\alpha(\rho_n\|\sigma_n)
 \ge\frac1{\alpha-1}\log(p_n^\alpha\beta_n^{1-\alpha})
 =\frac\alpha{\alpha-1}\log p_n-\log\beta_n.
\]
Rearranging yields
\begin{equation}\label{eq:app-testing}
 p_n\le
 2^{n\frac{\alpha-1}{\alpha}\Dt_\alpha^{\reg}(\N\|\M)}
 \beta_n^{\frac{\alpha-1}{\alpha}}.
\end{equation}
If $\beta_n=0$, finiteness of the final-state divergence forces $p_n=0$; if $p_n=0$, the bound is immediate.

\begin{corollary}[Strong converse and Stein's lemma for channel discrimination]\label{cor:discrimination}
For every $r>D^{\reg}(\N\|\M)$, there is $\kappa_r>0$ such that every $n$-use adaptive protocol whose type-II error is at most $2^{-nr}$ has success probability at most $2^{-n\kappa_r}$:
\begin{equation}\label{eq:SC}
 \beta_n\le2^{-nr}\quad\Longrightarrow\quad p_n\le2^{-n\kappa_r}.
\end{equation}
For $\epsilon\in(0,1)$, let $\beta_{n,\epsilon}$ be the infimum of the type-II errors among $n$-use adaptive protocols with $p_n\ge1-\epsilon$. Then
\begin{equation}\label{eq:app-stein}
 \lim_{n\to\infty}-\frac1n\log\beta_{n,\epsilon}
 =D^{\reg}(\N\|\M).
\end{equation}
\end{corollary}

\begin{proof}
For every $r>D^{\reg}(\N\|\M)$, Theorem~\ref{thm:main} guarantees that there exists an $\alpha>1$ with $\Dt_\alpha^{\reg}(\N\|\M)<r$. Set
$\kappa_r=\frac{\alpha-1}{\alpha}[r-\Dt_\alpha^{\reg}(\N\|\M)]>0$.
Substituting $\beta_n\le2^{-nr}$ into \eqref{eq:app-testing} proves \eqref{eq:SC}.

For the converse bound in \eqref{eq:app-stein}, \eqref{eq:app-testing} and $p_n\ge1-\epsilon$ give, uniformly over admissible adaptive protocols,
\[
 \beta_n\ge(1-\epsilon)^{\alpha/(\alpha-1)}
                2^{-n\Dt_\alpha^{\reg}(\N\|\M)}.
\]
This also bounds the infimum $\beta_{n,\epsilon}$, without requiring an optimal protocol. Therefore
\[
 -\frac1n\log\beta_{n,\epsilon}
 \le\Dt_\alpha^{\reg}(\N\|\M)
 +\frac\alpha{n(\alpha-1)}\log\frac1{1-\epsilon}.
\]
First take the upper limit as $n\to\infty$ at fixed $\alpha>1$, then let $\alpha\downarrow1$ using Theorem~\ref{thm:main}.

For achievability at $0<s<D^{\reg}(\N\|\M)$, \eqref{eq:regularized} supplies a block size $k$, a pure input on $R\otimes A^{\otimes k}$, and $s'>s$ such that its output states satisfy $D(\rho_k\|\sigma_k)>ks'$. Repeat this input $m$ times and apply state Stein's lemma \cite{HP,ON}: for large $m$, there are tests with acceptance tending to one and type-II error at most $2^{-mks'}$. For $n$ calls, take $m=\lfloor n/k\rfloor$ and ignore the outputs of the remaining calls. Since $mk/n\to1$, eventually $mks'\ge ns$, so these parallel tests satisfy $p_n\to1$ and $\beta_n\le2^{-ns}$. Their reference is $R^{\otimes m}$, whose dimension may grow with $n$. They are also valid adaptive protocols. Letting $s\uparrow D^{\reg}(\N\|\M)$ proves achievability; when $D^{\reg}(\N\|\M)=0$, nonnegativity of the error exponents suffices.
\end{proof}

The exponential strong converse \eqref{eq:SC} corresponds to \cite[Corollary 22 and Theorem 23]{FGW}. It also holds for parallel protocols by inclusion. The fixed-error limit holds for parallel protocols as well: the converse covers all adaptive tests, while the achievability tests above are parallel. This parallel limit is statement~(1) of \cite[Theorem 14]{Gour}.

\subsection{Subchannel AEP and approximation}
A \emph{subchannel} $\mathcal G:\cL(A)\to\cL(B)$ is a completely positive map that is trace nonincreasing: $\Tr\mathcal G(X)\le\Tr X$ for every $X\ge0$. Equivalently, $\mathcal G^*(I_B)\le I_A$. 
Subchannel AEP is the following statement:
\begin{corollary}[Subchannel AEP]\label{cor:subchannel-aep}
Let $\N,\M$ be channels. For every fixed $\epsilon\in(0,1)$, smoothing $\N^{\otimes n}$ over nearby subchannels has asymptotic max-relative entropy $D^{\reg}(\N\|\M)$:
\begin{equation}\label{eq:subchannel-aep}
 \lim_{n\to\infty}\frac1n
 \inf_{\substack{\mathcal G_n\ge_{\cp}0,\ \mathcal G_n^*(I)\le I\\
                 \norm{\mathcal G_n-\N^{\otimes n}}_\diamond\le\epsilon}}
 D_{\max}(\mathcal G_n\|\M^{\otimes n})
 =D^{\reg}(\N\|\M).
\end{equation}
\end{corollary}

Requiring the smoother to preserve trace can invalidate the AEP \cite[Section VI]{Gour}, so the subchannel AEP is the natural alternative. Cf. statement~(3) of \cite[Theorem 14]{Gour}.\footnote{We use the unhalved diamond norm, whereas Gour uses generalized diamond distance, which is at most the unhalved diamond norm and bounds the decrease of a test's acceptance probability \cite[Eqs.~(31)--(32), (211)]{Gour}, so the same argument covers both metrics.}

We prove \eqref{eq:subchannel-aep} using the following stronger uniform approximation, obtained from continuity, the R\'enyi testing bound [Eq.~\eqref{eq:high} of Lemma~\ref{lem:testing}], and Lemma~\ref{lem:filter}.
\begin{corollary}[Exponential subchannel approximation]\label{cor:subchannel}
For every $r>D^{\reg}(\N\|\M)$, there are $c_r>0$ and subchannels
$\mathcal G_n:\cL(A^{\otimes n})\to\cL(B^{\otimes n})$, for $n\ge1$, such that
\begin{equation}\label{eq:dominated}
 \mathcal G_n\le_{\cp}2^{nr}\M^{\otimes n},\qquad
 \norm{\mathcal G_n-\N^{\otimes n}}_\diamond\le4\,2^{-nc_r}.
\end{equation}
\end{corollary}
\begin{proof}
Choose $\alpha>1$ with $\Dt_\alpha^{\reg}(\N\|\M)<r$, and let
$c_r=\tfrac12(\alpha-1)[r-\Dt_\alpha^{\reg}(\N\|\M)]$ and $\eta_n=2^{-nc_r}$.
By Lemma~\ref{lem:testing}~\eqref{eq:high} and Lemma~\ref{lem:filter}, a Stinespring isometry $V_n$ of $\N^{\otimes n}$ admits $W_n$ in the same environment with
\[
 \Phi_{W_n}\le_{\cp}2^{nr}\M^{\otimes n},\qquad
 \norm{V_n-W_n}_\infty\le\eta_n.
\]
Since $V_n$ is an isometry, $\norm{W_n}_\infty\le1+\eta_n$. Define $W_n'=W_n/(1+\eta_n)$ and $\mathcal G_n=\Phi_{W_n'}$. Then $\norm{W_n'}_\infty\le1$, so
\[
 \mathcal G_n^*(I_{B^{\otimes n}})=W_n'^\dagger W_n'\le I_{A^{\otimes n}}.
\]
Thus $\mathcal G_n$ is a subchannel. Rescaling also preserves the required domination:
\[
 \mathcal G_n=\frac{\Phi_{W_n}}{(1+\eta_n)^2}
 \le_{\cp}\frac{2^{nr}}{(1+\eta_n)^2}\M^{\otimes n}
 \le_{\cp}2^{nr}\M^{\otimes n}.
\]
For the approximation error, the triangle inequality gives
\[
 \norm{V_n-W_n'}_\infty
 \le\norm{V_n-W_n}_\infty+\norm{W_n-W_n'}_\infty
 \le\eta_n+\frac{\eta_n}{1+\eta_n}\norm{W_n}_\infty
 \le2\eta_n.
\]
For operators $U,V$ with the same environment,
\[
 \Phi_U-\Phi_V
 =\Tr_E\!\left[(U-V)(\cdot)U^\dagger+V(\cdot)(U-V)^\dagger\right].
\]
The trace-norm triangle inequality, also after tensoring with an arbitrary reference, gives the dilation estimate \cite{KSW}
\[
 \norm{\Phi_U-\Phi_V}_\diamond
 \le(\norm U_\infty+\norm V_\infty)\norm{U-V}_\infty
\]
with $U=W_n'$ and $V=V_n$. Since $\Phi_{V_n}=\N^{\otimes n}$, this gives
\[
 \norm{\mathcal G_n-\N^{\otimes n}}_\diamond
 \le(\norm{W_n'}_\infty+\norm{V_n}_\infty)\norm{W_n'-V_n}_\infty
 \le(1+1)\,2\eta_n=4\,2^{-nc_r},
\]
proving \eqref{eq:dominated}.
\end{proof}

\begin{proof}[Proof of Corollary~\ref{cor:subchannel-aep}]
For the upper bound in \eqref{eq:subchannel-aep}, fix $r>D^{\reg}(\N\|\M)$. Corollary~\ref{cor:subchannel} gives feasible maps for all sufficiently large $n$, since $4\,2^{-nc_r}\le\epsilon$, with $D_{\max}(\mathcal G_n\|\M^{\otimes n})\le nr$. Take the upper limit and then let $r\downarrow D^{\reg}(\N\|\M)$.

For the converse, let $(\psi_n,T_n)$ be the parallel input and test from the block construction proving \eqref{eq:app-stein} at $0<s<D^{\reg}(\N\|\M)$. The diamond-distance bound limits the change of this test's acceptance probability to $\epsilon$: pair the output difference with $0\le T_n\le I$ and bound it by its trace norm. Thus every feasible smoother and every $\lambda>0$ with $\mathcal G_n\le_{\cp}\lambda\M^{\otimes n}$ satisfy
\[
 p_n-\epsilon
 \le\Tr[T_n(\id\otimes\mathcal G_n)(\psi_n)]
 \le\lambda\beta_n\le\lambda2^{-ns}.
\]
Since $p_n-\epsilon\to1-\epsilon>0$, its max-relative entropy is at least $ns+\log(p_n-\epsilon)$ for large $n$, and the logarithmic term divided by $n$ vanishes. Taking the lower limit and then $s\uparrow D^{\reg}(\N\|\M)$ proves the lower bound. When $D^{\reg}(\N\|\M)=0$, use any input state $\psi$ to obtain
\[
 1-\epsilon\le\Tr[(\id\otimes\mathcal G_n)(\psi)]\le\lambda.
\]
Thus $D_{\max}(\mathcal G_n\|\M^{\otimes n})\ge\log(1-\epsilon)$, whose normalized limit is zero.
\end{proof}

\subsection{Sharp testing threshold}
The hockey-stick divergence has a direct interpretation as a parallel binary test at a rate $r\ge0$. Choose a pure joint input $\psi_n$ on $R\otimes A^{\otimes n}$, with $R\simeq A^{\otimes n}$, and a test $0\le T_n\le I$ on $R\otimes B^{\otimes n}$. Let $p_n$ and $\beta_n$ be the probabilities that this test accepts $\N$ when the channel is $\N$ and $\M$, respectively. By \eqref{eq:hchannel},
\[
 E_{2^{nr}}(\N^{\otimes n}\|\M^{\otimes n})
 =\max_{\psi_n,T_n}\bigl(p_n-2^{nr}\beta_n\bigr).
\]
The factor $2^{nr}$ penalizes false acceptance under $\M$ at rate $r$. If this maximum approaches one, the optimizing tests have $p_n\to1$ and $2^{nr}\beta_n\to0$; if it approaches zero, no parallel test retains a positive limiting payoff. The following zero-one law follows immediately from our results.

\begin{corollary}[Sharp parallel testing threshold]\label{cor:sharp-testing}
For every $r\ge0$ with $r\ne D^{\reg}(\N\|\M)$,\footnote{No assertion is made at the threshold itself.}
\begin{equation}\label{eq:app-sharp}
 \lim_{n\to\infty}E_{2^{nr}}(\N^{\otimes n}\|\M^{\otimes n})
 =\begin{cases}
   1,&0\le r<D^{\reg}(\N\|\M),\\
   0,&r>D^{\reg}(\N\|\M).
  \end{cases}
\end{equation}
\end{corollary}

\begin{proof}
For $0\le r<D^{\reg}(\N\|\M)$, choose $s$ strictly between them. The parallel block tests constructed for \eqref{eq:app-stein}, as in \cite[Lemma 12]{Gour}, satisfy $p_n\to1$ and $\beta_n\le2^{-ns}$, so
\[
 1\ge E_{2^{nr}}(\N^{\otimes n}\|\M^{\otimes n})
 \ge p_n-2^{-n(s-r)}\longrightarrow1.
\]
For $r>D^{\reg}(\N\|\M)$, Theorem~\ref{thm:main} supplies $\alpha>1$ with $\Dt_\alpha^{\reg}(\N\|\M)<r$. Lemma~\ref{lem:testing}~\eqref{eq:high} then gives
\[
 0\le E_{2^{nr}}(\N^{\otimes n}\|\M^{\otimes n})
 \le2^{-n(\alpha-1)[r-\Dt_\alpha^{\reg}(\N\|\M)]}\longrightarrow0.
\]
\end{proof}

The above-threshold assertion is statement~(4) of \cite[Theorem 14]{Gour}, here with exponential decay; the full off-threshold law is \cite[Eq.~(209)]{Gour}.

Gour's equivalence theorem also gives the channel Lorenz AEP from this same above-threshold testing condition \cite[Theorem 14, (4) implies (2)]{Gour}, which we shall not discuss in detail.\par\medskip

\noindent\textbf{Acknowledgments.}
We would like to thank Tony Metger, Renato Renner and David Sutter for initial collaborations. We thank David Sutter and Lukas Schmitt for coordinating on arXiv submission. YY acknowledges the support of the National Natural Science Foundation of China via the Excellent Young Scientists Fund (Hong Kong and Macau) Project 12322516 and the support of Research Grants Council (RGC) of Hong Kong SAR via the NSFC/RGC Joint Research Scheme N\_HKU7107/24.

\end{document}